\documentclass{scrartcl}

\pdfoutput=1

\usepackage{stmaryrd}
\usepackage{jon-tikz}
\usepackage{maths/maths}
\usepackage{lattices/lattices}
\usepackage{categories/categories}
\usepackage{kripkesem/kripkesem}
\usepackage{macros}
\usepackage{mathpartir}
\usepackage{epigraph}
\usepackage{preamble-online}

\usepackage[numbers]{natbib}
\usetikzlibrary{positioning}
\usetikzlibrary{calc}
\usetikzlibrary{arrows}
\usetikzlibrary{decorations.markings}
\newcommand*{\StrikeThruDistance}{0.15cm}%
\tikzset{strike thru arrow/.style={
    decoration={markings, mark=at position 0.5 with {
        \draw [thick,-] 
            ++ (-\StrikeThruDistance,-\StrikeThruDistance) 
            -- ( \StrikeThruDistance, \StrikeThruDistance);}
    },
    postaction={decorate},
}}

\title{Two-dimensional Kripke Semantics I: Presheaves}
\author{G. A. Kavvos}
\date{February 2024}

\begin{document}

\maketitle

\begin{abstract}
  The study of modal logic has witnessed tremendous development following the
  introduction of Kripke semantics. However, recent developments in programming
  languages and type theory have led to a second way of studying modalities,
  namely through their categorical semantics. We show how the two correspond. 
\end{abstract}

\section{Introduction}
 

The development of modal logic has undergone many phases
\cite{chagrov_1996,blackburn_2001,goldblatt_2006,van_benthem_2010}. It is
widely accepted that one of the most important developments was the relational
semantics of Kripke \cite{kripke_1963,kripke_1963b,kripke_1965} \cite[\S
1]{blackburn_2001} \cite[\S 4.8]{goldblatt_2006}. Kripke semantics has proven
time and again that it is intuitive and technically malleable, thereby exerting
sustained influence over Computer Science.

However, over the last 30 years another way of studying modalities has evolved:
looking at modal logic through the prism of the \emph{Curry-Howard-Lambek
correspondence} \cite{lambek_1988,sorensen_2006,wadler_2015} yields new
computational intuitions, often with surprising applications in both programming
languages and formal proof. The tools of the trade here are type theory and
category theory.

Up to now these two ways of looking at modalities have been discussed in
isolation. The purpose of this paper is to establish a connection: I will show
that the Kripke and categorical semantics of modal logic are part of a
\emph{duality}. It is well-known that dualities between Kripke and algebraic
semantics exist: the \emph{J\'{o}nsson-Tarski duality} is one of the
cornerstones of classical modal logic \cite[\S 5]{blackburn_2001}. The main
contribution of this paper is to show that such dualities can be elevated to the
level of \emph{proofs}. The punchline is that a \emph{profunctor} $R : \Op{\CC}
\times \CC \to \SET$, considered as a \emph{proof-relevant relation} on a
category $\CC$, uniquely corresponds to a categorical model of modal logic on
the category of presheaves on $\CC$.

There are two obstacles to overcome to get to that result. The first is that we
must work over an \emph{intuitionistic} substrate: most research on types and
categories is forced to do so, for unavoidable reasons \cite[\S 8]{lambek_1988}.
We must therefore first develop a duality for \emph{intuitionistic modal logic}.
However, there is no consensus on what intuitionistic modal logic is! The
problem is particularly acute in the presence of $\lozenge$ \cite{das_2023}. I
will avoid this problem by making canonical choices at each step. First, I will
formulate a Kripke semantics based on \emph{bimodules}, i.e. relations that are
canonically compatible with a poset. Then, I will show how \emph{Kan extension}
uniquely determines two adjoint modalities, $\blacklozenge$ and $\Box$, from any
bimodule. The fact these arise automatically is evidence that they are the
canonical choice of intuitionistic modalities.

The second obstacle stems from considering proofs. The jump from algebraic to
categorical semantics involves adding an extra `dimension' of proofs.
Consequently, in order to re-establish a duality, an additional dimension must
be added to Kripke semantics as well. I call the result a \emph{two-dimensional
Kripke semantics}. Category theorists will find it anticlimactic: it amounts to
the folklore observation that a proof-relevant Kripke semantics is essentially a
semantics in a presheaf category. 

Indeed, a sizeable proportion of this paper consists of folklore results that
are well-known to experts. However, many of them are drawn from related but
distinct areas: logic, order theory, category theory, and topos theory. As a
result, it does not appear that all of them are well-known by a \emph{single}
expert. Thus, the synthesis presented here appears to be new.

The results I present in this paper show that there are deep connections between
modal logic and presheaf categories. This is important, as the latter are
ubiquitous in logic and related fields: presheaf models are used in fields as
disparate as categorical homotopy theory \cite{riehl_2014,cisinski_2019}, type
theory \cite{hofmann_1997}, concurrency
\cite{joyal_1996,cattani_1997,cattani_2005}, memory allocation
\cite{oles_1985,oles_1997}, synthetic guarded domain theory
\cite{birkedal_2012}, second-order syntax and algebraic theories
\cite{fiore_1999,hamana_2004,fiore_2008,fiore_2010,fiore_2010b}, higher-order
abstract syntax \cite{hofmann_1999}, and so on. As a result, the connections
presented here may enable synthetic reasoning \emph{via} modalities in a variety
of logical settings.

In \cref{section:intuitionistic-logic-1} I recall the Kripke and algebraic
semantics of intuitionistic logic, and outline the duality between Kripke
semantics and certain complete Heyting algebras, the \emph{prime algebraic
lattices}. Then I extend this duality to intuitionistic modal logic in
\cref{section:modal-logic-1} by showing how a relation that is compatible with
the intuitionistic order---a bimodule---gives rise to modalities through Kan
extension. In \cref{section:intuitionistic-logic-2} I add proofs to
intuitionistic logic, and elevate the duality to one between `two-dimensional
frames' and presheaf categories. I then repeat this exercise for intuitionistic
modal logic in \cref{section:modal-logic-2} by promoting bimodules to
profunctors on the relational side, and adding an adjunction on the categorical
side.

For general background in orders please refer to the book by Davey and Priestley
\cite{davey_2002}. Given a poset $(D, \sqsubseteq_D)$ let the \emph{opposite}
poset $\Op{D}$ be given by reversing the partial order; that is, $x
\sqsubseteq_{\Op{D}} y$ iff $y \sqsubseteq_D x$.
A \emph{lattice} has all finite meets and joins. A \emph{complete lattice} has
arbitrary ones. A complete lattice is \emph{infinitely distributive} just if the
law $a \land \bigvee_{i} b_i = \bigvee_{i} a \land b_i$ holds. Such lattices are
variously called \emph{frames}, \emph{locales}, or \emph{complete Heyting
algebras} \cite{johnstone_1982,mac_lane_1994,picado_2012}.

\section{Intuitionistic Logic I}
  \label{section:intuitionistic-logic-1}

There are many types of semantics for intuitionistic logic, including Kripke,
Beth, topological, and algebraic semantics. Bezhanishvili and Holliday
\cite{bezhanishvili_2019} argue that these form a strict hierarchy, with Kripke
being the least general, and algebraic the most general. I will briefly review
the elements of these extreme points of the spectrum.


The Kripke semantics of intuitionistic logic are given by \emph{Kripke frames},
i.e. partially-ordered sets $(W, \sqsubseteq)$ \cite[\S 2.2]{chagrov_1996}. $W$
is referred to as the set of \emph{possible worlds}, and $\sqsubseteq$ as the
\emph{information order}.  A world $w \in W$ is a `state of knowledge,' and $w
\sqsubseteq v$ means that moving from $w$ to $v$ potentially entails an increase
in the amount of information.

Let $\Up*{W}$ be the set of \emph{upper sets} of $W$, i.e. the subsets $S
\subseteq W$ such that $w \in S$ and $w \sqsubseteq v$ implies $v \in S$. A
\emph{Kripke model} $\Model{M} = (W, \sqsubseteq, V)$ consists of a Kripke frame
$(W, \sqsubseteq)$ as well as a function $V : \Vars \to \Up*{W}$. The
\emph{valuation} $V$ assigns to each propositional variable $p \in \Vars$ an
upper set $V(p) \subseteq W$, which is the set of worlds in which $p$ is true.
The idea is that, once a proposition becomes true, it must remain true as
information increases.

We are now able to inductively define a relation $\KSat[\Model{M}]{w}{\varphi}$
with the meaning that $\varphi$ is true in world $w$ of model $\Model{M}$. The
only interesting clause is that for implication:
\[
  \KSat[\Model{M}]{w}{\varphi \to \psi}\
  \defequiv\
    \forall w \sqsubseteq v.\
    \KSat[\Model{M}]{v}{\varphi} \text{ implies } \KSat[\Model{M}]{v}{\psi}
\]
This definition is famously monotonic:
  if $\KSat[\Model{M}]{w}{\varphi}$ and $w \sqsubseteq v$ then
  $\KSat[\Model{M}]{v}{\varphi}$.
Kripke semantics is sound and complete for intuitionistic logic 
\cite{chagrov_1996,blackburn_2001}.


The algebraic semantics of intuitionistic logic consist of \emph{Heyting
algebras}. These are lattices such that every map $- \land x : L \to L$ has a
right adjoint, i.e. for $x, y \in L$ there is an element $\HeyExp{x}{y} \in L$
such that $c \land x \sqsubseteq y$ iff $c \sqsubseteq \HeyExp{x}{y}$. Such
lattices are always distributive. Assuming that one has an interpretation
$\sem{p} \in L$ of each proposition $p$, each formula $\varphi$ of
intuitionistic logic is inductively mapped to an element $\sem{\varphi} \in L$
using the corresponding algebraic structure.
I will not expound further on Heyting algebras; see \cite[\S
7.3]{chagrov_1996} \cite[\S 1.1]{borceux_1994c} \cite[\S I.8]{mac_lane_1994}.
But I note that they are sound and complete for intuitionistic logic.

\subsection{Prime algebraic lattices}
  \label{section:prime-algebraic-lattices}

Let $(W, \sqsubseteq)$ be any Kripke frame, and let $\TV{} \defeq \{ 0
\sqsubseteq 1 \}$. Consider the poset $\FUNC{W}{\TV}$ of monotonic functions
from $W$ to $\TV$, ordered pointwise. This poset has a number of curious
properties.

First, the monotonicity of $p : W \to \TV$ means that if $p(w) = 1$ and $w
\sqsubseteq v$, then $p(v) = 1$. Hence, the subset $U \defeq \ReIdx{p}{1}$ of
$W$ is an upper set. Conversely, every upper set $U \subseteq W$ gives rise to a
monotonic $p_U : W \to \TV$ by setting $p_U(w) = 1$ if $w \in U$, and $0$
otherwise. Consequently, there is an \emph{order bijection}
\[
  \Up*{W} \cong \FUNC{W}{\TV}
\]
with the order on $\Up*{W}$ being inclusion. I will liberally treat upper sets
and elements of $\FUNC{W}{\TV}$ as the same, but prefer the latter notation for
reasons that will become clear later.

Second, given any $w \in W$, consider its \emph{principal upper set} $\UpSet{w}
\defeq \SetComp{v \in W}{w \sqsubseteq v} \in \FUNC{W}{\TV}$. This set consists
of worlds with potentially more information than that found in world $w$. A
simple argument shows that $w \sqsubseteq v$ iff $\UpSet{v} \subseteq
\UpSet{w}$.\footnote{This is an order-theoretic consequence of the Yoneda
lemma.} Thus, this gives an \emph{order embedding}
\[
  \UpSet : \Op{W} \to \FUNC{W}{\TV}
\]
which can be shown to preserve meets and exponentials.

Third, the poset $\FUNC{W}{\TV}$ is a \emph{complete lattice}: arbitrary joins
and meets are given pointwise. Viewing the elements of $\FUNC{W}{\TV}$ as upper
sets, these joins and meets correspond to arbitrary unions and intersections of
upper sets, which are also upper. Moreover, this lattice satisfies the infinite
distributive law, so it is a \emph{complete Heyting algebra}---synonymously a
\emph{frame} or \emph{locale} \cite{johnstone_1982,picado_2012}. Given two upper
sets $X, Y \subseteq W$ their exponential is given by \cite[\S
1.9]{de_jongh_1966}
\[ 
  \HeyExp{X}{Y} \defeq 
  \SetComp{ w \in W } { \forall w \sqsubseteq v.\ v \in X \text{ implies } v \in Y }
\]

Fourth, the principal upper sets $\UpSet{w}$ are special, in that they are
\emph{prime}.\footnote{Such elements are variously called \emph{completely
  join-irreducible} \cite{raney_1952}, \emph{supercompact}
  \cite{banaschewski_1988} \cite[\S VII.8]{picado_2012}, \emph{completely
  (join-)prime} 
\cite{winskel_2009}, or simply \emph{join-prime} \cite[\S 1.3]{gehrke_2024}.}
An element $d$ of a complete lattice $L$ is \emph{prime} just if
\[
  d \sqsubseteq \Sup X\ \text{ implies }\ \exists x \in X.\ d \sqsubseteq x
\]
This says that $d$ contains a tiny, indivisible fragment of information: as soon
as it approximates a supremum, it must approximate something in the set that is
being upper-bounded. The prime elements of $\FUNC{W}{\TV}$ are exactly the
principal upper sets $\UpSet{w}$ for some $w \in W$.

Fifth, the complete lattice $\FUNC{W}{\TV}$ is \emph{prime algebraic}. This
means that all its elements can be reconstructed by `multiplying' or `sticking
together' prime elements. In symbols, a complete lattice $L$ is prime algebraic
whenever for every element $d \in L$ we have \[ d = \Sup \SetComp{ p \in L }{p
\sqsubseteq d, p \text{ prime }} \] Such lattices are variously called
\emph{completely distributive, algebraic lattices} \cite[\S 10.29]{davey_2002}
or \emph{superalgebraic lattices} \cite[\S VII.8]{picado_2012}. In fact, it can
be shown that any such lattice is essentially of the form $\FUNC{W}{\TV}$, i.e.
a lattice of upper sets; this was shown by Raney in the 1950s \cite{raney_1952},
and independently by Nielsen, Plotkin and Winskel in the 1980s
\cite{nielsen_1981}. See the paper by Winskel for the use of prime algebraic
lattices in semantics \cite{winskel_2009}.

Finally, the fact every element can be reconstructed as a supremum of primes
means that it is possible to canonically extend any monotonic $f : W \to W'$ to
a monotonic $\FUNC{\Op{W}}{\TV} \to W'$, as long as $W'$ is a complete lattice.
Diagrammatically, in the situation
\begin{equation}
  \label{diagram:left-kan-extension-poset}
  \begin{tikzpicture}[node distance=2.5cm, on grid, baseline=(current  bounding  box.center)]
    \node (C) {$W$};
    \node (PSHC) [right = 4cm of C] {$\FUNC{\Op{W}}{\TV}$};
    \node (E) [below = 2cm of PSHC] {$W'$};
    \path[->] (C) edge node [above] {$\UpSet$} (PSHC);
    \path[->] (C) edge node [below = .5em] {$f$} (E);
    \path[->, dashed] (PSHC) 
      edge 
        node [left] {$\LKan{f}$} 
        node [right = .15em] {$\Adjoint$}
      (E);
    \path[->, bend right = 45, dotted] (E) edge node [right] {$\NerveS{f}$} (PSHC);
  \end{tikzpicture}
\end{equation}
there exists a unique $\LKan{f}$ which preserves all joins and satisfies
$\LKan{f}\prn{\UpSet w} = f(w)$. It is given by
\[
  \LKan{f}\prn{S} \defeq \Sup \SetComp{ f(w) }{ w \in S }
\]
$\LKan{f}$ is called the \emph{(left) Kan extension} of $f$ along $\UpSet$. As
$\LKan{f}$ preserves all joins and $\FUNC{W}{\TV}$ is complete it has a right
adjoint $\NerveS{f}$, by the adjoint functor theorem \cite[\S 7.34]{davey_2002}
\cite[\S I.4.2]{johnstone_1982}. For any complete lattice $W'$ this situation
amounts to a bijection
\[
  \Hom[\POSET]{W}{W'} \cong \Hom[\CSEMILATT]{\FUNC{\Op{W}}{\TV}}{W'}
\]
where $\CSEMILATT$ is the category of complete lattices and join-preserving
maps.

Suppose then that we have a Kripke model $(W, \sqsubseteq, V)$. The construction
given above induces a Heyting algebra $\FUNC{W}{\TV}$. Defining $\sem{p} \defeq 
V(p)$ we obtain an algebraic model of intuitionistic logic, which interprets
every formula $\varphi$ as an upper set $\sem{\varphi} \in \FUNC{W}{\TV}$. This
is the upper set of worlds in which a formula is true \cite[Theorem
7.20]{chagrov_1996}:
\begin{theorem}
  \label{theorem:kripke-vs-algebraic-1}
  $\KSat{w}{\varphi}$ if and only if $w \in \sem{\varphi}$.
\end{theorem}
%
Thus, every Kripke semantics corresponds to a prime algebraic lattice.


\begin{remark}
  This shows that a Kripke semantics is a particular kind of algebraic
  semantics. Thus, we can deduce the completeness of the latter from the
  completeness of the former: if a formula is valid in all Heyting algebras, it
  must be valid in all prime algebraic lattices, and hence valid in all Kripke
  semantics. If the Kripke semantics is complete, then the formula must be
  provable. Therefore, the algebraic semantics is then complete as well.

  The opposite direction---viz.~proving the completeness of Kripke semantics
  from completeness of the algebraic semantics---cannot be shown constructively.
  The reason is that it requires the construction of \emph{prime filters}, which
  is a weak form of choice. I will investigate the details of this mismatch in a
  sequel paper.
\end{remark}

\subsection{Morphisms}


The simplest kind of morphism between Kripke frames is a \emph{monotonic} map $f
: W \to W'$. Frames and monotonic maps form the category $\POSET$ of posets.
Given a monotonic $f : W \to W'$ we can define a monotonic $\Pre{f} :
\FUNC{W'}{\TV} \to \FUNC{W}{\TV}$ by taking $p : W' \to \TV$ to $p \circ f : W
\to \TV$.
Viewing the elements of $\FUNC{W'}{\TV}$ as upper sets, $\Pre{f}$ maps the upper
set $S \subseteq W'$ to the set $\SetComp{ v \in W' }{ f(v) \in S } \subseteq
W$, which is upper by the monotonicity of $f$. $\Pre{f}$ preserves arbitrary
joins and meets. It is thus the morphism part of a functor
$
  \FUNC{-}{\TV} : \Op{\POSET} \fto \PRIMEALGLATT
$
to the category $\PRIMEALGLATT$ of prime algebraic lattices and complete lattice
homomorphisms. 

Moreover, the functor $\FUNC{-}{\TV}$ is an equivalence! By the adjoint functor
theorem any complete lattice homomorphism $\Pre{f} : L' \to L$ has a left and
right adjoint:
\begin{diagram}
  \label{diagram:essential-poset}
  \begin{tikzpicture}[node distance=2.5cm, on grid, baseline=(current  bounding  box.center)]
    \node (W) {$L$};
    \node (V) [right = 4cm of W] {$L'$};
    \path[->] (V) edge node [above, near start] {$\Pre{f}$} (W);
    \path[->, bend left=30, dashed] (W) edge node [above] {$\RKan{f}$} (V);
    \path[->, bend right=30, dashed] (W) edge node [below] {$\LKan{f}$} (V);
    \node (Adj1) [above right = 0.35cm and 2 cm of W] {$\AdjointLDown$};
    \node (Adj2) [below right = 0.35cm and 2 cm of W] {$\AdjointLDown$};
  \end{tikzpicture}
\end{diagram}
Given a prime algebraic lattice $L$, let $\Primes{L} \subseteq L$ be the
sub-poset of prime elements. It can be shown that the left adjoint $\LKan{f}$
maps primes to primes \cite[Lemma 1.23]{gehrke_2024}. We can thus restrict it to
a function $\Primes{L} \to \Primes{L'}$. This defines a functor $\Primes{-} :
\PRIMEALGLATT \fto \Op{\POSET}$ with the property that $\Primes{\FUNC{W}{\TV}}
\cong W$. All in all, this amounts to a \emph{duality}
\begin{equation}
  \label{equation:primealglatt-duality}
  \Op{\POSET} \Equiv \PRIMEALGLATT
\end{equation}


However, monotonic maps are not particularly well-behaved from the perspective
of logic, as they do not preserve nor reflect `local' truth. 
This is the privilege of \emph{open maps}.

\begin{definition}
  \label{definition:open-poset}
  Let $i_0 : \mathbb{1} \to \TV$ map the unique point of $\mathbb{1}
  \defeq \{ \ast \}$ to $0 \in \TV$.
  A monotonic map $f : W \to W'$ of Kripke frames is \emph{open} just when it has
  the right lifting property with respect to $i_0 : \mathbb{1} \to \TV$, i.e.
  when every commuting diagram of the form
  \[
    \begin{tikzpicture}[diagram]
      \SpliceDiagramSquare{
        height = 1.2cm,
        width = 3cm,
        nw = \mathbb{1},
        sw = \TV,
        west = i_0,
        ne = W,
        se = W',
        east = f,
      }
      \path[->, dotted] (sw) edge (ne);
    \end{tikzpicture}
  \]
  in $\POSET$ has a diagonal filler (dashed) that makes it commute.
\end{definition}
In other words, $f$ is open if whenever $f(w) \sqsubseteq v'$ there exists a $w'
\in W$ with $w \sqsubseteq w'$ and $f(w') = v'$.\footnote{Such morphisms are
  often called \emph{p-morphisms} \cite[\S 2.3]{chagrov_1996} or \emph{bounded
  morphisms} \cite[\S 2.1]{blackburn_2001}. According to Goldblatt
  \cite{goldblatt_2006} open maps were introduced by de Jongh and Troelstra
  \cite{de_jongh_1966} in intuitionistic logic, and by Segerberg
  \cite{segerberg_1968} in modal logic. More rarely they are called
  \emph{functional simulations}, and led us to bisimulations \cite[\S
  3.2]{sangiorgi_2009}. The name is chosen because such maps are open with
  respect to the \emph{Alexandrov topology} on a poset, whose open sets are the
upper sets \cite[\S 1.8]{johnstone_1982}.} Open maps send upper sets to upper
sets \cite[Prop. 2.13]{chagrov_1996}. Thus
\begin{lemma}
  \label{lemma:kripke-open-truth}
  Let $\Model{M} = (W, \sqsubseteq, V)$ and $\Model{N} = (W', \sqsubseteq, V')$
  be Kripke models, and $f : W \to W'$ be open. Suppose $V = f^{-1} \circ V'$,
  i.e. $w \in V(p)$ iff $f(w) \in V'(p)$. Then $\KSat[\Model{M}]{w}{\varphi}$
  iff $\KSat[\Model{N}]{f(w)}{\varphi}$.
\end{lemma}
Write $\KVal{W}{\varphi}$ to mean that $\KSat[(W, \sqsubseteq, V)]{w}{\varphi}$
for any valuation $V$ and $w \in W$. Then
\begin{lemma}
  \label{lemma:kripke-open-surj-truth}
  If $f : W \to W'$ is open and surjective, then $\KVal{W}{\varphi}$ implies
  $\KVal{W'}{\varphi}$.
\end{lemma}
Recall now the induced map $\Pre{f} : \FUNC{W'}{\TV} \to \FUNC{W}{\TV}$ for a
monotonic $f : W \to W'$. The following lemma allows us to characterise the
openness and surjectivity of $f$ in terms of $\Pre{f}$.
\begin{lemma} \hfill
  \label{lemma:fstar-open}
  \begin{enumerate}
    \item 
      $f : W \to W'$ is open iff $\Pre{f} : \FUNC{W'}{\TV} \to \FUNC{W}{\TV}$
      preserves exponentials.
    \item
      $f : W \to W'$ is surjective iff $\Pre{f} : \FUNC{W'}{\TV} \to
      \FUNC{W}{\TV}$ is injective.
  \end{enumerate}
\end{lemma}

Consequently, the duality \eqref{equation:primealglatt-duality} may be
restricted to two wide subcategories:
\begin{align}
  \label{equation:posetopen-primealglattcc}
  \Op{\POSET}_\text{open} \Equiv \PRIMEALGLATTCC
  & \quad &
  \Op{\POSET}_\text{open,surj} \Equiv \PRIMEALGLATTCCINJ
\end{align}
The morphisms of the categories to the left of $\Equiv$ are open (resp. and
surjective) maps, and the morphisms of the categories to its right are
\emph{complete Heyting homomorphisms}, i.e. complete lattice homomorphisms that
preserve exponentials (resp. and are injective).

Finally, let us consider the classical case---as a sanity-check. This amounts to
restricting $\POSET$ to its subcategory of discrete orders, i.e. $\SET$. In this
case every map is open. The corresponding restriction on the other side is to
the category $\CABA$ of \emph{complete atomic Boolean algebras}, yielding the
usual \emph{Tarski duality} $\Op{\SET} \Equiv \CABA$ \cite{kishida_2018}.

\subsection{Related work}

The origins of the construction of a Heyting algebra from a Kripke frame seems
to be lost in the mists of time. The earliest occurrence I have located is in
the book by Fitting \cite[\S 1.6]{fitting_1969}, where it is attributed to an
exercise in the book by Beth \cite{beth_1959}. 

The duality \eqref{equation:primealglatt-duality} appears to be
folklore---folklore enough to be included as an exercise in new textbooks
\cite[Ex. 1.3.10]{gehrke_2024}; see also Ern\'{e} \cite{erne_1991}. However, I
have not been able to find any mention of the dualities of
\eqref{equation:posetopen-primealglattcc} in the literature.

Both the dualities \eqref{equation:primealglatt-duality} and
\eqref{equation:posetopen-primealglattcc} involve just prime algebraic lattices,
which is a far cry from encompassing all Heyting algebras. It is possible to do
so, by enlarging the category $\POSET$ to a class of ordered topological spaces
called \emph{descriptive frames} \cite[\S 8.4]{chagrov_1996}. The resulting
duality is called \emph{Esakia duality} \cite{esakia_2019} \cite[\S
4.6]{gehrke_2024} \cite[\S 2.3]{bezhanishvili_2006}.

A survey on dualities for classical modal logic is given by Kishida
\cite{kishida_2018}.

\section{Modal Logic I}
  \label{section:modal-logic-1}

The task now is to extend the results of \cref{section:intuitionistic-logic-1}
to \emph{intuitionistic modal logic}.

There is disagreement on what a minimal intuitionistic modal logic is. This
arises no matter the methodology we choose---be it relational, algebraic, or
proof-theoretic. The situation becomes even more complex if we include a diamond
modality ($\lozenge$): see Das and Marin \cite{das_2023} and Wolter and
Zakharyaschev \cite{wolter_1999b} for a discussion.

In this paper I will adopt the \emph{intuitionistic propositional logic with
Galois connections} of Dzik, J\"{a}rvinen, and Kondo \cite{dzik_2010}, for
reasons that will become clear in a moment. This extends intuitionistic logic
with modalities $\blacklozenge$ and $\Box$, and the two inference rules
\begin{mathpar}
  \inferrule{
    \blacklozenge \varphi \to \psi
  }{
    \varphi \to \Box \psi
  }
  \and
  \text{and}
  \and
  \inferrule{
    \varphi \to \Box \psi
  }{
    \blacklozenge \varphi \to \psi
  }
\end{mathpar}
These rules correspond to a \emph{Galois connection} \cite[\S
7.23]{davey_2002}, i.e. an adjunction $\blacklozenge \Adjoint \Box$ between
posets. They imply the derivability of the following rules, amongst others
\cite[Prop. 2.1]{dzik_2010}.
\begin{mathpar}
  \inferrule{
    \varphi \to \psi
  }{
    \Box \varphi \to \Box \psi
  }
  \and
  \inferrule{
    \varphi
  }{
    \Box \varphi
  }
  \and
  \inferrule{
  }{
    \Box \top
  }
  \and
  \inferrule{
    \blacklozenge \bot
  }{
    \bot
  }
  \and
  \inferrule{
    \varphi \to \psi
  }{
    \blacklozenge \varphi \to \blacklozenge \psi
  }
  \and
  \inferrule{
  }{
    \blacklozenge (\varphi \lor \psi) \leftrightarrow \blacklozenge \varphi \lor
    \blacklozenge \psi
  }
  \and
  \inferrule{
  }{
    \Box (\varphi \land \psi) \leftrightarrow \Box \varphi \land
    \Box \psi
  }
\end{mathpar}
The notation of the `black diamond' modality $\blacklozenge$ may appear unusual.
However, I will argue that this logic is, in a way, the canonical intuitionistic
modal logic.


The Kripke semantics of classical modal logic is given by a \emph{modal frame}
$(W, R)$, which consists of a set $W$ and an arbitrary \emph{accessibility
relation} $R \subseteq W \times W$ \cite[\S 1]{blackburn_2001}. If the same set
of worlds $W$ is already part of an intuitionistic Kripke frame $(W,
\sqsubseteq)$ we must take care to ensure that $\sqsubseteq$ and $R$ are
\emph{compatible}. There are many compatibility conditions that one can consider
\cite{plotkin_1986} \cite[\S 3.3]{simpson_1994}. However, I will take a hint
from the category theory literature, and seek a canonical definition of what it
means for a relation to be compatible with a poset.

Recall that relations can be presented as functions $R : W \times W \to \TV$
which map a pair of worlds $(w, v)$ to $1$ iff $w \mathbin{R} v$. I will ask
that $R$ is such function, but with a twist:
\begin{definition}
  A \emph{bimodule} $R : W_1 \pfto W_2$ 
  is a monotonic map $R : \Op{W_1} \times W_2 \to \TV$.
\end{definition}
Thus, a relation $R \subseteq W_1 \times W_2$ is a bimodule just if $w'
\sqsubseteq w \mathbin{R} v \sqsubseteq v'$ implies $w' \mathbin{R} v'$. This
means that $R$ can absorb changes in information on either side: contravariantly
on the first component, and covariantly on the second. This is a standard,
minimal way to define what it means to be `a relation in $\POSET$.' It is
strongly reminiscent of bimodules in abstract algebra.

We can then define a \emph{modal Kripke frame} $(W, \sqsubseteq, R)$ to be a
Kripke frame $(W, \sqsubseteq)$ equipped with a bimodule $R : W \pfto W$. A
\emph{modal Kripke model} $\Model{M} = (W, \sqsubseteq, R, V)$
adds to this a function $V : \Vars \to \Up*{W}$. We extend
$\KSat[\Model{M}]{w}{\varphi}$ to modal formulae:
\begin{align*}
  \KSat[\Model{M}]{w}{\blacklozenge \varphi}\
  &\defequiv\
    \exists v.\ v \mathbin{R} w \text{ and } \KSat[\Model{M}]{v}{\varphi}
  \\
  \KSat[\Model{M}]{w}{\Box \varphi}\
  &\defequiv\
    \forall v.\ w \mathbin{R} v \text{ implies } \KSat[\Model{M}]{v}{\varphi}
\end{align*}
There are a number of things to note about this definition. First, there is a
deep duality between the clauses: not only do we exchange $\forall$ for
$\exists$, but we also flip the variance of $R$. As a result, $\blacklozenge$
uses the relation in the \emph{opposite variance} to the more traditional
$\lozenge$ modality (hence the change in notation). Second, the clause for the
$\Box$ modality is the traditional one; some streams of work on intuitionistic
modal logic adopt a slightly different one \cite{plotkin_1986,simpson_1994},
which is equivalent to this in the presence of the bimodule condition. Finally,
this definition is monotonic: the bimodule conditions on $R$ suffice to show
that
  if $\KSat[\Model{M}]{w}{\varphi}$ and $w
  \sqsubseteq v$ then $\KSat[\Model{M}]{v}{\varphi}$.
Dzik et al. \cite[\S 5]{dzik_2010} prove that this semantics is sound and complete.


The algebraic semantics of this logic is given by a Heyting algebra $H$ equipped
with two monotonic maps $\blacklozenge, \Box : H \to H$ which form an adjunction
$\blacklozenge \Adjoint \Box$, i.e. a Galois connection. Dzik et al. \cite[\S
4]{dzik_2010} prove that this semantics is also sound and complete.


We are now in a position to relate the Kripke and algebraic semantics of this
intuitionistic modal logic. Let $(W, \sqsubseteq, R)$ be a modal Kripke frame,
and consider the map $\lambda R : \Op{W} \to \FUNC{W}{\TV}$ obtained by
cartesian closure of $\POSET$. This map takes $w \in W$ to the upper set
$\SetComp{ v \in W }{ w \mathbin{R} v }$ of worlds accessible from $w$. Putting
$\lambda R$ in \eqref{diagram:left-kan-extension-poset} we obtain through Kan
extension the diagram
\begin{equation}
  \begin{tikzpicture}[node distance=2.5cm, on grid, baseline=(current bounding box.center)]
    \node (C) {$\Op{W}$};
    \node (PSHC) [right = 4cm of C] {$\FUNC{W}{\TV}$};
    \node (E) [below = 2cm of PSHC] {$\FUNC{W}{\TV}$};
    \path[->] (C) edge node [above] {$\UpSet$} (PSHC);
    \path[->] (C) edge node [below = .5em] {$\lambda R$} (E);
    \path[->, dashed] (PSHC) 
      edge 
        node [left] {$\blacklozenge_R$} 
        node [right = .15em] {$\Adjoint$}
      (E);
    \path[->, bend right = 45, dotted] (E) edge node [right] {$\Box_R$} (PSHC);
  \end{tikzpicture}
\end{equation}
where we write $\blacklozenge_R$ for $\LKan{\lambda R}$ and $\Box_R$ for
$\NerveS{\lambda R}$. It can be shown that these maps are given by
\begin{align*}
  \blacklozenge_R(S)
    &\defeq \SetComp{ w \in W }{ \exists v.\ v \mathbin{R} w \text{ and } v \in S }
    \\
  \Box_R(S)
    &\defeq \SetComp{ w \in W }{ \forall v.\ w \mathbin{R} v \text{ implies } v \in S }
\end{align*}
Thus, any bimodule $R$ defines an adjunction $\blacklozenge_R \Adjoint \Box_R$
on $\FUNC{W}{\TV}$. Correspondingly, any adjunction $\blacklozenge \Adjoint
\Box$ on $\FUNC{W}{\TV}$ yields a monotonic map $\blacklozenge \circ \UpSet*{-}
: \Op{W} \to \FUNC{W}{\TV}$, which uniquely corresponds to a bimodule $\Op{W}
\times W \to \TV$ by the cartesian closure of $\POSET$. 

Thus, starting from a bimodule, i.e. a relation that is compatible with the
information order, we have canonically obtained a model of intuitionistic modal
logic on $\FUNC{W}{\TV}$ through Kan extension: $\FUNC{W}{\TV}$ is a complete
Heyting algebra, and we define $\sem{\blacklozenge \varphi} = \blacklozenge_R
\sem{\varphi}$ and $\sem{\Box \varphi} = \Box_R \sem{\varphi}$. We immediately
obtain a modal analogue to Theorem \ref{theorem:kripke-vs-algebraic-1}:
\begin{theorem}
  \label{theorem:kripke-vs-algebraic-2}
  For any modal formula $\varphi$, $\KSat{w}{\varphi}$ if and only if $w \in
  \sem{\varphi}$.
\end{theorem}

\subsection{Morphisms}

We define a category $\BIMOD$ with bimodules $R : W_1 \pfto W_2$ as objects. A
\emph{bimodule morphism} from $R : W_1 \pfto W_2$ to $R' : W'_1 \pfto W'_2$ is a
pair $(f, g)$ of monotonic maps $f : W_1 \to W'_1$ and $g : W_2 \to W'_2$ such
that $R(w, v) \sqsubseteq R'(f(w), g(v))$. Stated in terms of relations, it must
be that $w \mathbin{R} v$ implies $f(w) \mathbin{R'} g(v)$. 

We define the subcategory $\EBIMOD$ to consist of (endo)bimodules $R : W \pfto
W$ and pairs of maps $(f, f)$. Thus, objects are bimodules on a single poset
$W$, and morphisms are monotonic maps $f : W \to W'$ that preserve the relation,
i.e. $w \mathbin{R} v$ implies $f(w) \mathbin{R} f(v)$. In other words, the
objects of $\EBIMOD$ are modal Kripke frames, and the morphisms are monotonic,
relation-preserving maps.

Recall the adjunctions and modalities induced by a monotonic $f : W \to W'$:
\begin{equation}
  \label{diagram:essential-uppers-modal}
  \begin{tikzpicture}[node distance=2.5cm, on grid, baseline=(current  bounding  box.center)]
    \node (W) {$\FUNC{W}{\TV}$};
    \node (V) [right = 4cm of W] {$\FUNC{W'}{\TV}$};
    \path[->] (V) edge node [above, near start] {$\Pre{f}$} (W);
    \path[->, bend left=30, dashed] (W) edge node [above] {$\RKan{f}$} (V);
    \path[->, bend right=30, dashed] (W) edge node [below] {$\LKan{f}$} (V);
    \node (Adj1) [above right = 0.35cm and 2 cm of W] {$\AdjointLDown$};
    \node (Adj2) [below right = 0.35cm and 2 cm of W] {$\AdjointLDown$};
    \path[->] (W) edge [out=135, in=225, loop] node [left] {$\Box_{R}$} (W);
    \path[->] (V) edge [out=45, in=315, loop] node [right] {$\Box_{R'}$} (V);
  \end{tikzpicture}
\end{equation}

\begin{lemma}
  \label{lemma:morphism-bimod-iff-inclusion}
  $f : W \to W'$ is a morphism of bimodules $f : R \to R'$ iff
  $\Pre{f}\Box_{R'} \subseteq \Box_R \Pre{f}$.
\end{lemma}
This constitutes a duality
\begin{equation}
  \label{equation:bimod-primealglatto}
  \Op{\EBIMOD} \Equiv \PRIMEALGLATTO
\end{equation}
where $\PRIMEALGLATTO$ is the category with objects $(L, \Box_L)$, where $L$ is
a prime algebraic lattice and $\Box_L : L \to L$ is an operator that preserves
all meets. By the adjoint functor theorem, such operators always have a left
adjoint $\blacklozenge_L : L \to L$. Thus, this category contains algebraic
models of intuitionistic modal logic (but not all of them). By the preceding
section each such adjunction corresponds uniquely to a bimodule. The morphisms
of $\PRIMEALGLATTO$ are complete lattice homomorphisms $h : L \to L'$ such that
$h \Box_L \sqsubseteq \Box_{L'} h$. By the preceding lemma they correspond
precisely to morphisms of bimodules.

However, as with monotone maps, morphisms of bimodules do not preserve local
truth; for that we need a notion of \emph{modally open} maps. 

\begin{definition}
  \label{definition:modally-open}
  Let $(W, \sqsubseteq, R)$ and $(W', \sqsubseteq, R')$ be modal Kripke frames.
  A bimodule morphism $f : R \to R'$ is \emph{modally open} just if whenever
  $f(w) \mathbin{R'} v$ then there exists a $w' \in W$ with $w \mathbin{R} w'$
  and $f(w') \sqsubseteq v$.
\end{definition}
This is similar to \cref{definition:open-poset}, but ever so slightly weaker:
instead of requiring $f(w') = v'$, it requires that the information in $f(w')$
can be increased to $v'$. Like \cref{definition:open-poset}, it can also be
written homotopy-theoretically, but that requires some ideas from double
categories that are beyond the scope of this paper. We have the analogous result
about preservation of truth:
\begin{lemma}
  \label{lemma:kripke-open-modal-truth}
  Let $\Model{M} = (W, \sqsubseteq, R, V)$ and $\Model{N} = (W', \sqsubseteq,
  R', V')$ be modal Kripke models, $f : W \to W'$ be open and modally open, and
  $V = f^{-1} \circ V'$. Then $\KSat[\Model{M}]{w}{\varphi}$
  iff~$\KSat[\Model{N}]{f(w)}{\varphi}$.
\end{lemma}
\begin{lemma}
  Let $f : W \to W'$ be open, modally open, and surjective. If
  $\KVal{W}{\varphi}$ then~$\KVal{W'}{\varphi}$.
\end{lemma}

The following result relates the modal openness of $f$ to $\Pre{f}$.
\begin{lemma}
  \label{lemma:kripke-modally-open-iff}
  $f : R \to R'$ is modally open iff $\Box_R \Pre{f} = \Pre{f}
  \Box_{R'}$ iff $\LKan{f} \blacklozenge_R = \blacklozenge_{R'} \LKan{f}$.
\end{lemma}
Thus, the duality \eqref{equation:bimod-primealglatto} may be restricted to
dualities between wide subcategories:
\begin{align}
  \label{equation:bimodopen-primealglattoo}
  \EBIMOD^\text{op}_\text{moo} \Equiv \PRIMEALGLATTO_{\Rightarrow{}o}
  & \quad &
  \EBIMOD^\text{op}_\text{moo, surj} \Equiv \PRIMEALGLATTO_{\Rightarrow{}o, \text{inj}}
\end{align}
The morphisms to the left of $\Equiv$ are open and modally open (resp. and
surjective); and the to the right of it preserve exponentials and commute with
operators (resp. and are injective).

Let us consider the restriction of this duality to the classical setting---as a
sanity check. A bimodule on a discrete poset is just a relation on a set. The
corresponding restriction on the right is to CABAs with operators, and complete
homomorphisms which commute with operators. We thus obtain the \emph{Thomason
duality} $\textbf{MFrm}^\text{op}_\text{open} \Equiv \textbf{CABAO}$ between
Kripke frames and modally open maps on the left, and CABAs with operators to the
right \cite{thomason_1975,kishida_2018}.

\subsection{Related work}


Many works have presented a Kripke semantics for intuitionistic modal logic. All
such semantics assume two accessibility relations: a preorder for the
intuitionistic dimension, and a second relation for the modal dimension. What
varies is their \emph{compatibility conditions}.

The first work to present such a semantics appears to be that of Fischer Servi
\cite{fischer_servi_1980}. One of the required compatibility conditions is
$(\sqsubseteq) \circ R \subseteq R \circ (\sqsubseteq)$. This is weaker than
having a bimodule, but sufficient to prove soundness.

The first work to recognise the importance of bimodules was Sotirov's 1979
thesis. His results are summarised in a conference abstract \cite[\S
4]{sotirov_1984}: they include the completeness of a minimal intuitionistic
modal logic with a $\Box$, the $\text{K}$ axiom, and the necessitation rule.
Bo\v{z}i\'{c} and Do\v{s}en \cite{bozic_1984} repeat the study for the same
logic, but for a semantics based on the Fischer Servi compatibility conditions.
However, they note that their completeness proof actually constructs half a
bimodule (a `condensed' relation). They also point out that bimodules, which
they call `strictly condensed' relations, are sound and complete for their
logic. Wolter and Zakharyaschev \cite[\S 2]{wolter_1997} argue that bimodule and
Fischer Servi semantics are equi-expressive.

Plotkin and Stirling \cite{plotkin_1986} attempt to systematise the Kripke
semantics of intuitionistic modal logic. Their frame conditions allow
`transporting a modal relation upwards' along any potential increases of
information on either side. This paper and all its descendants---notably the
thesis of Simpson \cite[\S 3.3]{simpson_1994}---adopt a different satisfaction
clause for $\Box$ which uses both $\sqsubseteq$ and $R$. In the presence of the
bimodule conditions this satisfaction clause is equivalent to the classical one,
which I use here.


The bimodule condition and the \emph{complex algebra} construction (or fragments
thereof) have made scattered appearances in the literature: in the early work of
Sotirov \cite{sotirov_1984} and Bo\v{z}i\'{c} and Do\v{s}en \cite{bozic_1984};
in the work of Wolter and Zakharyaschev
\cite{wolter_1999,wolter_1997,wolter_1999b}, Hasimoto \cite[\S
4]{hasimoto_2001}, and Or\l{}owska and Rewitzky \cite{orlowska_2007}; and of
course in Dzik et al. \cite[\S 7]{dzik_2010}. With the exception of the last
one, none of these references discuss the $\blacklozenge$ modality. Moreover, in
none of these references are the categorical aspects of this construction
discussed. 


As mentioned before, dualities between frames and algebras have played a
significant role in modal logic.
Thomason \cite{thomason_1975} and Goldblatt \cite{goldblatt_1974} also
considered morphisms of frames, respectively obtaining \emph{Thomason duality}
and (categorical) \emph{J\'{o}nsson-Tarski duality} between descriptive frames
and Boolean Algebras with Operators (BAOs) \cite[\S 6.5]{goldblatt_2006}.
Kishida \cite{kishida_2018} surveys a number of dualities for classical modal
logic.

The duality 
\eqref{equation:bimod-primealglatto} is stated by Gehrke \cite[Thm.
2.5]{gehrke_2016} who attributes it to J\'{o}nnson \cite{jonnson_1952}, even
though no such theorem appears in that paper.

The dualities of 
\eqref{equation:bimodopen-primealglattoo} are
the direct intuitionistic analogues to that of Thomason. I have not been able to
find them anywhere in the literature.


According to the extensive survey of Menni and Smith \cite{menni_2014}, the idea
that the commonly-used modalities $\Box$ and $\lozenge$ are often part of
adjunctions $\blacklozenge \Adjoint \Box$ and $\lozenge \Adjoint \blacksquare$
is implicitly present throughout the development of modal logic. However, these
were not made explicit in a logic until the 2010s, when they appeared in the
work of Dzik et al. \cite{dzik_2010} and Sadrzadeh and Dyckhoff
\cite{sadrzadeh_2010}. The same perspective plays a central r\^{o}le in the
exposition of Kishida \cite{kishida_2018}. 

The $\blacklozenge$ modality has appeared before in \emph{tense logics} as a
`past' modality \cite{ewald_1986,goranko_2023}.

\section{Intuitionistic Logic II}
  \label{section:intuitionistic-logic-2}

In the rest of this paper we will \emph{categorify} \cite{baez_1998} the notion
of Kripke semantics. The main idea is to replace posets by categories, so that
the order $w \sqsubseteq v$ is replaced by a morphism $w \to v$. As there might
be multiple morphisms $w \to v$, this allows the recording of not just the fact
$v$ may signify more information than $w$, but also the \emph{manner} in which
it does so. The reflexivity and transitivity of the poset are then replaced by
the identity and composition laws of the category. This adds a dimension of
\emph{proof-relevance} to Kripke semantics.

A corresponding change in our algebraic viewpoint will be that of replacing the
set $\TV$ of truth values with the category $\SET$. This is a classic Lawverean
move \cite{lawvere_1973}. Notice that this is lopsided, as is usual in
intuitionistic logic: while the falsity $0$ is only represented by one value,
viz. the empty set, the truth $1$ can be represented by any non-empty set $X$.
The elements of $X$ can be thought of as a \emph{proofs} of a true statement.

Let us then trade the frame $(W, \sqsubseteq)$ for an arbitrary category $\CC$.
It remains to define what it means to have a \emph{proof} that the formula
$\varphi$ holds at a world $w \in \CC$. We denote the set of all such proofs by
$\sem{\varphi}_w$. Assuming we are given a set $\sem{p}_w$ for each proposition
$p$ and world $w$, here is a first attempt:
\begin{gather*}
  \begin{align*}
    \sem{\bot}_w &\defeq \emptyset &
    \sem{\top}_w &\defeq \{ \ast \} &
    \sem{\varphi \land \psi}_w &\defeq \sem{\varphi}_w \times \sem{\psi}_w &
    \sem{\varphi \lor \psi}_w &\defeq \sem{\varphi}_w + \sem{\psi}_w
  \end{align*} \\
  \sem{\varphi \to \psi}_w \defeq 
    \DepFun{v}{\CC}{\Hom[\CC]{w}{v} \to \sem{\varphi}_v \to \sem{\psi}_v}
\end{gather*}
where for a family of sets $\DelimPrn{B_a}_{a \in A}$ we let
\[
  \DepFun{a}{A}{B_a}  \defeq
  \SetComp{ f : A \to \bigcup_{a \in A} B_a \DelimMin{3} }
            { \forall a \in A. f(a) \in B_a }
\]
This closely follows the usual Kripke semantics, but adds proofs. For example, a
proof in $\sem{\varphi_1 \land \varphi_2}_w$ is a pair $(x, y)$ of a proof $x
\in \sem{\varphi_1}_w$ and a proof $y \in \sem{\varphi_2}_w$. Similarly, a proof
$F \in \sem{\varphi \to \psi}_w$ is a function which maps a proof of `increase
in information' $f : w \to v$ to a function $F(v)(f) : \sem{\varphi}_v \to
\sem{\psi}_v$. In turn, this function maps proofs in $\sem{\varphi}_v$ to proofs
in $\sem{\psi}_v$.

To show that this definition is monotonic we have to demonstrate it on proofs:
given a proof $x \in \sem{\varphi}_w$ and a morphism $f : w \to v$ we have to
define a proof $f \cdot x \in \sem{\varphi}_v$. Assuming that we are given this
operation for propositions, we can extend it by induction; e.g.
\[
  \begin{array}{lllll}
    f \cdot (x, y)
        &\defeq& (f \cdot x, f \cdot y) 
        &\in \sem{\varphi \land \psi}_v \\
    f \cdot F
        &\defeq& (z : \CC) \mapsto
                  (g : \Hom[\CC]{v}{z}) \mapsto
                  (x : \sem{\varphi}_z) \mapsto F(z)(g \circ f)(x)
        &\in \sem{\varphi \to \psi}_v
  \end{array}
\]
Moreover, this definition is compatible with $\CC$, in the sense that $g \cdot
(f \cdot x) = (g \circ f) \cdot x$ and $\IdArr{w} \cdot x = x$. We thus obtain a
(covariant) \emph{presheaf}
$
  \sem{\varphi} : \CC \fto \SET
$
for each formula $\varphi$.

It is well-known that the proofs of intuitionistic logic form a
\emph{bicartesian closed category} (biCCC), i.e. a category with finite
(co)products and exponentials \cite{lambek_1980}. A biCCC can be seen as a
categorification of a Heyting algebra: formulae are objects of the category, and
proofs are morphisms. We will not expound on this further; see
\cite{lambek_1988,crole_1993,awodey_2010}.

It should therefore be the case that the semantics described above form a biCCC.
Indeed, it is a well-known fact of topos theory that the \emph{category of
presheaves} $\FUNC{\CC}{\SET}$ is a biCCC. In fact, the construction of
exponentials \cite[\S I.6]{mac_lane_1994} reveals that our definition above is
deficient: we should restrict $\sem{\varphi \to \psi}_w$ to contain only those
functions $F$ that satisfy a \emph{naturality condition}, i.e. those which for
any $f : w \to v_1$, $g : v_1 \to v_2$, and $x \in \sem{\varphi}_{v_1}$ satisfy
\[
  g \cdot F(v_1)(f)(x) = F(v_2)(g \circ f)(g \cdot x)
\]

From this point onwards I will identify two-dimensional Kripke semantics with
categorical semantics in a category of presheaves $\FUNC{\CC}{\SET}$.

\subsection{Presheaf categories}

The category $\FUNC{\CC}{\SET}$ of covariant presheaves is eerily similar to
prime algebraic lattices. In a sense they are just the same; but, having traded
$\TV$ for $\SET$, they have become proof-relevant.

First, letting $P \in \FUNC{\CC}{\SET}$, an element $x \in P(w)$ is a proof that
$P$ holds at a `world' $w \in \CC$. A morphism $f : w \to v$ of $\CC$ then leads
to a proof $f \cdot x \defeq P(f)(x) \in P(v)$ that $P$ holds at $v$. Thus, the
presheaf $P$ is very much like an upper set.

Second, the \emph{representable presheaves} $\Yo{w} \defeq \Hom[\CC]{w}{-} : \CC
\to \SET$ are the proof-relevant analogues of the principal upper set. By the
Yoneda lemma they constitute an \emph{embedding}
\[
  \Yo : \Op{\CC} \fto \FUNC{\CC}{\SET}
\]
which moreover preserves limits and exponentials \cite{awodey_2010}.

Third, the category $\FUNC{\CC}{\SET}$ is both complete and cocomplete, with
limits and colimits computed pointwise \cite[\S I]{mac_lane_1994}. It is also
`distributive' in an appropriate sense \cite[\S 3.3]{anel_2021}, which makes it
into a \emph{Grothendieck topos}. It is thus a cartesian closed category, with
exponential
\[
  (\HeyExp{P}{Q})(w) \defeq \Hom{P \times \Yo{w}}{Q}
\]
which is essentially the two-dimensional semantics of implication I gave above.

Fourth, the representables $\Yo{w}$ are special, in that
they are \emph{tiny} \cite{yetter_1987}.
\begin{definition}
  An object $d \in \DD$ is \emph{tiny} just if $\Hom{d}{-} : \DD \to
  \SET$ preserves colimits.\footnote{In the literature this property is often
  referred to as \emph{external tininess} (cf.~internal tininess).}
\end{definition}
Tininess is a proof-relevant version of primality: it implies that for any $f :
w \to \Colim[i] v_i$ there exists an $i$ such that $f$ is equal to the
composition of a morphism $w \to v_i$ with the injection $v_i \to \Colim[i]
v_i$. By the Yoneda lemma it follows that all representables $\Yo{w}$ are tiny,
as they satisfy the above definition for $\DD \defeq \FUNC{\CC}{\SET}$ and $d
\defeq \Yo{w}$.

Fifth, the so-called \emph{co-Yoneda lemma} \cite[\S III.7]{mac_lane_1978}
shows that every $P \in \FUNC{\CC}{\SET}$ is a colimit of representables. This
means that it can be reconstructed by sticking together tiny elements:
\[
  P \cong \Colim[(w, x) \in \PSHEL{P}] \Yo{w}
\]
Like with prime algebraic lattices, there is a converse to this result: every
category which is generated by sticking together tiny elements is in fact a
presheaf category:
\begin{theorem}[Bunge \cite{bunge_1966}]
  \label{theorem:bunge}
  A category which is cocomplete and strongly generated by a small set of tiny
  objects is equivalent to $\FUNC{\CC}{\SET}$ for some small category $\CC$.
\end{theorem}
A textbook presentation of this result can be found in the book by Kelly
\cite[\S 5.5]{kelly_2005}.

Finally, the fact every element can be reconstructed as a colimit of
representables means that it is possible to uniquely extend any functor $f : \CC
\fto \DD$ to a cocontinuous functor $\FUNC{\Op{\CC}}{\TV} \fto \DD$, as long as
$\DD$ is cocomplete. Diagrammatically, in the situation
\begin{equation}
  \label{diagram:left-kan-extension}
  \begin{tikzpicture}[node distance=2.5cm, on grid, baseline=(current  bounding  box.center)]
    \node (C) {$\CC$};
    \node (PSHC) [right = 4cm of C] {$\FUNC{\Op{\CC}}{\SET}$};
    \node (E) [below = 2cm of PSHC] {$\DD$};
    \path[->] (C) edge node [above] {$\Yo$} (PSHC);
    \path[->] (C) edge node [below = .5em] {$f$} (E);
    \path[->, dashed] (PSHC) 
      edge 
        node [left] {$\LKan{f}$} 
        node [right = .15em] {$\Adjoint$}
      (E);
    \path[->, bend right = 45, dotted] (E) edge node [right] {$\NerveS{f}$} (PSHC);
  \end{tikzpicture}
\end{equation}
there exists an essentially unique cocontinuous $\LKan{f}$ with
$\LKan{f}\prn{\Yo{w}} = f(w)$. It is given by
\[
  \LKan{f}\prn{\Colim[(w, x) \in \PSHEL{P}] \Yo{w}} 
    \defeq \Colim[(w, x) \in \PSHEL{P}] f(w)
\]
$\LKan{f}$ is called the \emph{left Kan extension} of $f$ along $\Yo$. It has a
right adjoint $\NerveS{f}$ which is explicitly given by $\NerveS{f}(d) \defeq
\Hom{f(-)}{d}$. This amounts to an isomorphism
\[
  \Hom[\CAT]{\CC}{\DD} \cong \Hom[\COCONT]{\FUNC{\Op{\CC}}{\SET}}{\DD}
\]
where $\CAT$ is the category of categories, and $\COCONT$ is the category of
cocomplete categories and cocontinuous functors: see
\cite[Prop.~9.16]{awodey_2010} \cite[Cor.~6.2.6, Rem.~6.5.9]{riehl_2016} and
\cite[\S~X.3, Cor.~2]{mac_lane_1978} \cite[Th. 4.51]{kelly_2005}.

\begin{table}
  \caption{Categorification of Kripke semantics}
  \centering
  \label{table:categorification}
  \medskip
  \begin{tabular}{c|c}
    poset                   & category \\
    monotonic map           & functor \\
    upper sets              & presheaves \\
    principal upper set     & representable presheaf \\
    prime element           & tiny object \\
    prime algebraic lattice & presheaf category \\
    bimodule                & profunctor
\end{tabular}
\end{table}

All in all, presheaf categories are the categorification of prime algebraic
lattices.

\subsection{Cauchy-complete and spacelike categories}

Replacing posets with categories does not come for free: the extra dimension of
morphisms leads to situations that have no analogues in poset. Some of these are
problematic when thinking of $\CC$ as a two-dimensional Kripke frame. Perhaps
the most bizarre is the presence of \emph{idempotents}, i.e. morphisms $e : w
\to w$ with the property that $e \circ e = e$. Such morphisms represent a
non-trivial increase in information which confusingly leaves us in the same
world.

The presence of idempotents causes issues. For example, recall that, in prime
algebraic lattices, primes and principal upper sets coincide. The astute reader
will have noticed we did \emph{not} claim the analogous result in presheaf
categories: tiny objects are not necessarily representable in
$\FUNC{\CC}{\SET}$. For that, we need $\CC$ to be \emph{Cauchy-complete}
\cite{borceux_1986,borceux_1994a}.
\begin{definition}
  A category is \emph{Cauchy-complete} just if every idempotent splits, i.e. if
  every idempotent is equal to $s \circ r$ for a section-retraction pair $s$ and
  $r$.
\end{definition}
Note that every complete category is Cauchy-complete, including $\SET$ and
$\FUNC{\CC}{\SET}$.

This leads us to another troublesome situation, namely that of having
section-retraction pairs, i.e. $s : w \to v$ and $r : v \to w$ with $r \circ s =
\IdArr{w}$. In this case $w$ and $v$ contain no more information than each
other, but are not isomorphic. We may ask that this does not arise.
\begin{definition}
  A category satisfies the \emph{Hemelaer condition} \cite[Prop.
  5.8]{hemelaer_2022} just if every section-retraction pair is an isomorphism.
\end{definition}
Combining these two conditions is equivalent to the following definition.
\begin{definition}
  A category is \emph{spacelike} if every idempotent is an identity.
\end{definition}
%

Lawvere has identified this condition as having particular importance in
recognising petit toposes \cite{mclarty_2006}. We will not use it much, as it
restricts the dualities we wish to develop.

In the rest of this paper we will assume that our base categories $\CC$ are
Cauchy-complete, so that tiny objects coincide with representables.

%
%
%

\subsection{Morphisms}

The simplest kind of morphism between categories is a functor. Given a $f : \CC
\fto \DD$ we can define a functor $\Pre{f} : \FUNC{\DD}{\SET} \fto
\FUNC{\CC}{\SET}$ that takes $P : \DD \fto \SET$ to $P \circ f : \CC \fto \SET$.
This functor has left and right adjoints, which are given by Kan extension
\cite[A4.1.4]{johnstone_2002a}:
\begin{diagram}
  \label{diagram:essential-presheaf}
  \begin{tikzpicture}[node distance=2.5cm, on grid]
    \node (W) {$\FUNC{\CC}{\SET}$};
    \node (V) [right = 4cm of W] {$\FUNC{\DD}{\SET}$};
    \path[->] (V) edge node [above, near start] {$\Pre{f}$} (W);
    \path[->, bend left=30, dashed] (W) edge node [above] {$\RKan{f}$} (V);
    \path[->, bend right=30, dashed] (W) edge node [below] {$\LKan{f}$} (V);
    \node (Adj1) [above right = 0.35cm and 2 cm of W] {$\AdjointLDown$};
    \node (Adj2) [below right = 0.35cm and 2 cm of W] {$\AdjointLDown$};
  \end{tikzpicture}
\end{diagram}
Therefore $\Pre{f}$ preserves all limits and colimits, i.e. it is
(co)continuous. In short, the presheaf construction gives a functor
$\FUNC{-}{\SET} : \Op{\CATCC} \fto \PSHCAT$, where $\CATCC$ is the category of
small Cauchy-complete categories and functors; and $\PSHCAT$ is the category of
presheaf categories (over Cauchy-complete base categories) and (co)continuous
functors.

Moreover, this functor is an equivalence. Given a presheaf category we can
obtain its base as the subcategory of tiny objects
\cite[A1.1.10]{johnstone_2002a}. But how can we extract $f : \CC \fto \DD$ from
any (co)continuous functor $\Pre{f} : \FUNC{\DD}{\SET} \fto \FUNC{\CC}{\SET}$?
First, as presheaf categories are locally presentable, the adjoint functor
theorem implies that $\Pre{f}$ has left and right adjoints, as in
\eqref{diagram:essential-presheaf} \cite[\S 1.66]{adamek_1994}. This gives what
topos theorists call an \emph{essential geometric morphism}. Johnstone \cite[\S
A4.1.5]{johnstone_2002a} shows that every such morphism is induced by a $f : \CC
\fto \DD$, as $\LKan{f}$ preserves representables (when $\DD$ is
Cauchy-complete). We thus obtain a duality
\begin{equation}
  \label{duality:catcc-psh}
  \textbf{Cat}^\text{op}_\text{cc} \Equiv \PSHCAT
\end{equation}

As with posets, functors here fail to preserve truth; for that we need a notion
of openness.
\begin{definition}
  \label{definition:open-functor}
  $f : \CC \fto \DD$ is \emph{open} just if $\Pre{f} : \FUNC{\DD}{\SET} \fto
  \FUNC{\CC}{\SET}$ preserves exponentials.
\end{definition}
\begin{lemma}
  \label{lemma:kripke-open-truth-category}
  If $f : \CC \fto \DD$ is open then there is a natural isomorphism
  $\theta_w : \sem{\varphi}_w \cong \sem{\varphi}_{f(w)}$.
\end{lemma}
\cref{definition:open-functor} is somewhat underwhelming, as it does not give
explicit conditions that one can check---unlike \cref{definition:open-poset}.
However, obtaining such a description appears difficult.

Some information may be gleaned by considering $(\Pre{f}, \RKan{f}) :
\FUNC{\CC}{\SET} \fto \FUNC{\DD}{\SET}$ as a \emph{geometric morphism}. Such a
morphism is \emph{open} \cite{johnstone_1980} \cite[C3.1]{johnstone_2002b}
just if both the canonical maps $\Pre{f}(c \Rightarrow d) \to \Pre{f}(c)
\Rightarrow \Pre{f}(d)$ and $\Pre{f}(\Omega) \to \Omega$ are monic.
Johnstone \cite[C3.1]{johnstone_2002b} proves that $(\Pre{f}, \RKan{f})$ is open
iff for any $\beta : f(w) \to v'$ in $\DD$ there exists an $\alpha : w \to w'$
in $\CC$ and a section-retraction pair $s : v' \to f(w')$ and $r : f(w') \to v'$
with $s \circ \beta = f(\alpha)$. This superficially seems like a
categorification of \cref{definition:open-poset}. However, it only guarantees
that $\Pre{f}$ is \emph{sub-cartesian-closed}, whereas we need an isomorphism
for \cref{lemma:kripke-open-truth-category} to hold. 

A stronger condition is to ask that $(\Pre{f}, \RKan{f})$ be \emph{locally
connected}, i.e. that $\Pre{f}$ commute with dependent products
\cite[C3.3]{johnstone_2002b}. All such morphisms are open geometric morphisms.
This is stronger than what we need, but sufficient conditions on $f$ can be
given \cite[C3.3.8]{johnstone_2002b}.

Finally, an even stronger condition is to ask that $(\Pre{f}, \RKan{f})$ be
\emph{atomic}, i.e. that $\Pre{f}$ is a \emph{logical} functor. This means it
preserves exponentials and the subobject classifier \cite[A2.1,
C3.5]{johnstone_2002b}. All atomic geometric morphisms are locally connected.
This is again stronger than what we need, and a characterisation in terms of $f$
is elusive: see MathOverflow \cite{spivak_2016}.

It is easier to characterise when $(\Pre{f}, \RKan{f})$ is a \emph{surjective
geometric morphism}, i.e. when $\Pre{f}$ is faithful
\cite[A2.4.6]{johnstone_2002a}. This happens exactly when $f$ is
\emph{retractionally surjective}, i.e. whenever every $d \in \DD$ is the retract
of $f(c)$ for some $c \in \CC$ \cite[A2.4.7]{johnstone_2002a}. If $\DD$
satisfies the Hemelaer condition this reduces to $f$ being essentially
surjective.

Write $\KVal{\CC}{\varphi}$ to mean that $\sem{\varphi}_w$ is non-empty for
any $w \in \CC$ and any interpretation of $\sem{p}$.
\begin{lemma}
  Let $f : \CC \to \DD$ be open and retractionally surjective. If
  $\KVal{\CC}{\varphi}$~then~$\KVal{\DD}{\varphi}$.
\end{lemma}

We may thus restrict the duality \eqref{duality:catcc-psh} to dualities
\begin{align}
  \label{duality:catopen-pshsub}
  \textbf{Cat}^\text{op}_\text{cc, open} \Equiv \PSHCAT_{\Rightarrow}
  & \quad &
  \textbf{Cat}^\text{op}_\text{cc, open, rs} \Equiv \PSHCAT_{\Rightarrow, \text{f}}
\end{align}
In the first instance the category to the left of $\Equiv$ is that of small
Cauchy-complete categories and open functors; and to the right of $\Equiv$ it is
presheaf categories and (co)complete, cartesian closed functors. In the second
instance the category to the left of $\Equiv$ is that of small Cauchy-complete
categories and open, retractionally surjective functors; and to the right of
$\Equiv$ it is presheaf categories and (co)complete, faithful, cartesian closed
functors.

\section{Modal Logic II}
  \label{section:modal-logic-2}

To make a two-dimensional Kripke semantics for modal logic we have to
categorify relations. We took the first step by considering bimodules, i.e.
information-order-respecting relations. The second step can be taken by
replacing $\TV$ with $\SET$; this leads us to the notion of a relation between
categories, also known as a \emph{profunctor} or \emph{distributor}
\cite{benabou_2000} \cite[\S 7]{borceux_1994a}.

\begin{definition}
  A \emph{profunctor} $R : \CC \pfto \DD$ is a functor $R : \Op{\CC} \times \DD
  \to \SET$.
\end{definition}


To formulate a two-dimensional Kripke semantics for modal logic we replace modal
Kripke frames with a small Cauchy-complete category $\CC$ with an
(endo)profunctor $R : \Op{\CC} \times \CC \to \SET$. To obtain the modalities we
can now play the same trick: putting $\lambda R : \Op{\CC} \to \FUNC{\CC}{\SET}$
into \eqref{diagram:left-kan-extension} we canonically obtain the following
diagram by Kan extension:
\begin{equation}
  \begin{tikzpicture}[node distance=2.5cm, on grid, baseline=(current  bounding  box.center)]
    \node (C) {$\Op{\CC}$};
    \node (PSHC) [right = 4cm of C] {$\FUNC{\CC}{\SET}$};
    \node (E) [below = 2cm of PSHC] {$\FUNC{\CC}{\SET}$};
    \path[->] (C) edge node [above] {$\Yo$} (PSHC);
    \path[->] (C) edge node [below = .5em] {$\lambda R$} (E);
    \path[->, dashed] (PSHC) 
      edge 
        node [left] {$\blacklozenge_R$} 
        node [right = .15em] {$\Adjoint$}
      (E);
    \path[->, bend right = 45, dotted] (E) edge node [right] {$\Box_R$} (PSHC);
  \end{tikzpicture}
\end{equation}
Conversely, any adjunction $\blacklozenge \Adjoint \Box$ on $\FUNC{\CC}{\SET}$
corresponds to the (endo)profunctor on $\CC$ given by $(c_1, c_2) \mapsto
\Hom[\CC]{c_1}{\blacklozenge c_2}$.

We may then define $\sem{\blacklozenge \varphi} \defeq \blacklozenge_R
\sem{\varphi} : \CC \fto \SET$ and $\sem{\Box \varphi} \defeq \Box_R
\sem{\varphi} : \CC \fto \SET$. It is worth unfolding what a proof of $\Box
\varphi$ at a world $w$ is to obtain an explicit description:
\begin{equation}
  \label{equation:box-2d}
  \sem{\Box \varphi}_w
  = \prn{\Box_R \sem{\varphi}}(w)
  = \Hom[\FUNC{\CC}{\SET}]{\lambda R(w)}{\sem{\varphi}}
  = \Hom[\FUNC{\CC}{\SET}]{R(w, -)}{\sem{\varphi}}
\end{equation}
Thus, a proof that $\varphi$ holds at $w$ is a natural transformation $\alpha :
R(w, -) \To \sem{\varphi}$. This has the expected shape of Kripke semantics for
$\Box$: for each $v \in \CC$ and proof $x \in R(w, v)$ that $v$ is accessible
from $w$ it gives us a proof $\alpha_v(x) \in \sem{\varphi}_v$ that $\varphi$
holds at $v$.

It is a little harder to see what a proof of $\blacklozenge \varphi$ at a world
$w$ is. It becomes more perspicuous if we use the coend formula for the left Kan
extension \cite[\S 2.3]{loregian_2021}:
\begin{equation}
  \label{equation:blacklozenge-2d}
  \sem{\blacklozenge \varphi}
  = \LKan{\lambda R} \sem{\varphi}
  \cong \int^{v \in \CC} \Hom[\FUNC{\CC}{\SET}]{\Yo{v}}{\sem{\varphi}} \times \lambda R(v)
  \cong \int^{v \in \CC} \sem{\varphi}_v \times R(v, -)
\end{equation}
Hence, a proof that $\blacklozenge \varphi$ holds at $w$ consists of a world $v
\in \CC$, a proof that $R(v, w)$, and a proof that $\varphi$ holds at
$v$---which is exactly what one would  expect. The difference is that the coend
quotients some of these pairs according to the action of $\CC$. See Mac
Lane and Moerdijk \cite[\S VII.2]{mac_lane_1994} for a textbook exposition on
why this is a tensor product of $\sem{\varphi}$ and $\lambda R$.

How well does this fit the categorical semantics of modal logic? As with
intuitionistic modal logic, there is also a number of proposals of what that
might be. A fairly recent idea is to define it as the semantics of a
\emph{Fitch-style calculus}, as studied by Clouston \cite{clouston_2018}. This
is exactly a bicartesian closed category $\CC$ equipped with an adjunction:
\begin{diagram}
  \label{diagram:fitch-style}
  \begin{tikzpicture}[node distance=2.5cm, on grid]
    \node (W) {$\CC$};
    \node (V) [right = 4cm of W] {$\CC$};
    \path[->, bend left=20] (W) edge node [above] {$\Box$} (V);
    \path[->, bend left=20] (V) edge node [below] {$\blacklozenge$} (W);
    \node (Adj1) [right = 2cm of W] {$\AdjointLDown$};
  \end{tikzpicture}
\end{diagram}
The left adjoint $\blacklozenge$ is often written as lock. It does not commonly
appear as a modality, but as an operator on contexts that corresponds to
`opening a box' in Fitch-style natural deduction \cite[\S 5.4]{huth_2004}. The
modality $\Box$ is a right adjoint, so that it automatically preserves all
limits, including products. This idea has proven remarkably robust: variations
on it have worked well for modal dependent type theories
\cite{birkedal_2020,gratzer_2020,gratzer_2021,gratzer_2022,gratzer_2023}. The
fact that an adjunction on a presheaf category corresponds precisely to a
two-dimensional Kripke semantics is further evidence that this is the correct
notion of categorical model of modal logic. 

Finally, note that \eqref{equation:box-2d} and \eqref{equation:blacklozenge-2d}
look suspiciously similar to the modal structure of Normalization-by-Evaluation
models for modal type theories. This is explicitly visible in the paper by
Valliappan et al. \cite[\S 2]{valliappan_2022}, and also implicitly present in
the paper by Gratzer \cite{gratzer_2022b}.

\subsection{Morphisms}

Define the category $\PROF$ to have as objects profunctors. A morphism $(f, g,
\alpha) : R \to S$ from $R : \CC \pfto \DD$ to $S : \CC' \pfto \DD'$ consists of
functors $f : \CC \fto \CC'$ and $g : \DD \fto \DD'$, and a natural
transformation $\alpha : R(-, -) \To S\prn{f(-), g(-)}$. The subcategory
$\EPROF$ consists of endoprofunctors $R : \CC \pfto \CC$, and triples of the
form $(f, f, \alpha)$. I will synecdochically refer to $\alpha : R(-,-) \To
S(f(-), f(-))$ as a morphism of $\EPROF$.  Thus, objects are two-dimensional
Kripke frames, and morphisms are functors that proof-relevantly preserve the
relation.

\begin{lemma}
  \label{lemma:endoprofunctors-nat}
  Morphisms of endoprofunctors $\alpha : R(-, -) \To S(f(-), f(-))$ are in
  bijection with natural transformations $\gamma: \Pre{f} \Box_{S} \To \Box_{R}
  \Pre{f}$.
\end{lemma}
\begin{proof}
  Unfolding the definitions, $\gamma : \Hom{S(f(-), -)}{-} \To \Hom{R(-,
  -)}{\Pre{f}(-)}$. As $\LKan{f} \Adjoint \Pre{f}$  this is exactly a
  transformation $\Hom{S(f(-), -)}{-} \To \Hom{\LKan{f}R(-, -)}{-}$. By the
  Yoneda lemma, any such transformation arises by precomposition with a unique
  transformation $ \LKan{f}R(-,-) \To S(f(-), -)$. By $\LKan{f} \Adjoint
  \Pre{f}$ again, this uniquely corresponds to a transformation $\alpha : R(-,
  -) \To \Pre{f} S(f(-), -) = S(f(-), f(-))$.
\end{proof}
We thus obtain a duality
\begin{equation}
  \label{duality:eprof-pshcato}
  \EPROF^\text{op}_\text{cc} \Equiv \PSHCATO
\end{equation}
where $\PSHCATO$ is the category of presheaf categories $\FUNC{\CC}{\SET}$
equipped with a continuous $\Box : \FUNC{\CC}{\SET} \fto \FUNC{\CC}{\SET}$. Note
that, as presheaf categories are locally finitely presentable, $\Box$ always has
a left adjoint $\blacklozenge$. Thus, the objects are categorical models of
modal logic. Morphisms are pairs $(f, \gamma)$ of a (co)continuous $f : \CC \fto
\DD$ and a natural transformation $\gamma : \Pre{f} \Box \Rightarrow \Box
\Pre{f}$.

As before, open functors do not preserve truth; for that we need a notion of
modal openness. Let $\alpha : R(-, -) \To S(f(-), f(-))$. As pointed out in the
proof of \cref{lemma:endoprofunctors-nat} such an $\alpha$ uniquely corresponds
to a transformation $t_\alpha : \LKan{f} R(-, -) \To S(f(-), -)$. Its components 
\[
  t_{\alpha, c, v} : \int^{w \in V} R(c, w) \times \Hom[\DD]{f(w)}{v} \to S(f(c), v)
\]
map $x \in R(c, w)$ and $k : f(w) \to v$ to $S(\IdArr{f(c)}, k)(\alpha_{c,
v}(x))$. We can then say that

\begin{definition}
  \label{definition:modally-open-functor}
  $\alpha : R(-, -) \To S(f(-), f(-))$ is \emph{modally open} just if $t_\alpha$
  is an isomorphism.
\end{definition}
This asks that for every proof $y \in S(f(c), v)$ we should be able to find an
object $w \in \CC$, a proof $x \in R(c, w)$, and a morphism $k : f(w) \to v$, so
that $y = S(\IdArr{f(c)}, k)(\alpha_{c, v}(x))$. This is clearly a
categorification of \cref{definition:modally-open}, and leads to the following
lemma:

\begin{lemma}
  $\alpha$ is modally open iff the corresponding $\Pre{f} \Box_S \To \Box_R
  \Pre{f}$ is an isomorphism.
\end{lemma}
\begin{proof}
  The proof of \cref{lemma:endoprofunctors-nat} precomposes with $t_\alpha$ to
  get $\gamma$. Thus $\gamma$ is iso iff $t_\alpha$ is.
\end{proof}

Thus, the duality \eqref{duality:eprof-pshcato} may be restricted to dualities
between the wide subcategories
\begin{align}
  \label{equation:eprofopen-pshcato}
  \EPROF^\text{op}_\text{cc, moo} \Equiv \PSHCATO_{\Rightarrow o}
  & \quad &
  \EPROF^\text{op}_\text{cc, moo, rs} \Equiv \PSHCATO_{\Rightarrow of}
\end{align}
The morphisms to the left of $\Equiv$ are modally open, open maps (resp. and
retractionally surjective); and the morphisms to the right of $\Equiv$ are $(f,
\gamma)$ where $f$ is cartesian closed (resp. and faithful) and $\gamma :
\Pre{f} \Box \cong \Box \Pre{f}$ is a natural isomorphism.

\section{Other related work}

Perhaps the work most closely related to this paper is that on
\emph{Kripke-style lambda models} by Mitchell and Moggi \cite{mitchell_1991}.
These amount to elaborating the first-order definitions of applicative structure
and $\lambda$-model in the internal language of a presheaf category, with the
base category being a partial order. In practice this means that the
interpretation of function types is only a subfunctor of the exponential of
presheaves \cite[\S 8]{mitchell_1991}. However, Mitchell and Moggi prove that
these models are sound and complete for the $(\times\rightarrow)$ fragment, even
in the presence of empty types. They also use some general theorems about open
geometric morphisms to prove that any cartesian closed category can be presented
as such a model.

Another piece of work that bears kinship with the present one is Hermida's
fibrational account of relational modalities \cite{hermida_2011}. Hermida shows
that both the relational modalities $\lozenge$ and $\Box$ can be obtained
canonically as extensions of predicate logic to relations, with the modalities
arising as compositions of adjoints. The black diamond $\blacklozenge$ makes a
brief cameo as the induced left adjoint to $\Box$, as does the dual black box
\cite[\S 3.3]{hermida_2011}. While the decompositions obtained by Hermida seem
more refined than the results here, Kan extension does not make an explicit
appearance. As such, the relationship to the present work is yet to be
determined.

Awodey and Rabe \cite{awodey_2011} give a Kripke semantics for extensional
Martin-L\"{o}f type theory (MLTT), in which contexts are posets and types are
presheaves over them. They use topos-theoretic machinery to prove that every
locally cartesian closed category can be embedded in a presheaf category over a
poset; this result seems similar to one of Mitchell and Moggi, but the proof
appears entirely different. As a consequence, they show that presheaf categories
over posets form a complete class of models for extensional MLTT, in fact a
subclass of locally cartesian closed categories.

Alechina et al. \cite{alechina_2001} present dualities between Kripke and algebraic semantics
for constructive \textsf{S4} and propositional lax logic. Their interpretation
of $\Box$ follows that of Plotkin, Stirling and Simpson
\cite{plotkin_1986,simpson_1994}.

Ghilardi and Meloni \cite{ghilardi_1988} explore a presheaf-like interpretation
of (predicate) modal logic, which is similar to ours, albeit non-proof-relevant.
They work over the identity profunctor $\Hom{-}{-}$. They are hence forced to
weaken the definition of presheaf. See also \cite{ghilardi_1989,ghilardi_1991}.

Awodey, Kishida and Kotzsch \cite{awodey_2014} give a topos-theoretic semantics
for a higher-order version of intuitionistic \textsf{S4} modal logic. They also
briefly survey much previous work on presheaf-based and topos-theoretic
semantics for first-order modal logic. Their work is not proof-relevant.

Finally, there is clear methodological similarity between the results obtained
here and the results of Winskel and collaborators on open maps and bisimulation
\cite{joyal_1996,cattani_2005}. One central difference is that Winskel et al.
are mainly concerned with open maps between presheaves themselves, whereas I
only consider open maps between (two-dimensional) frames.


\section*{Acknowledgements}

I have benefitted significantly from conversations with Dan Licata, Nachiappan
Valliappan, Fabian Ruch, Amar Hadzihasanovic, Kohei Kishida, Sean Moss, Sam
Staton, Daniel Gratzer, Lars Birkedal, Jonathan Sterling, Philip Saville, and
Gordon Plotkin.

This work was supported by the UKRI Engineering and Physical Sciences
Research Council grants EP/Y000242/1, EP/Y033418/1, the UKRI International
Science Partnerships Fund (ISPF), and a Royal Society Research Grant.

\bibliography{2dks}

\end{document}